\documentclass[runningheads]{llncs}
\usepackage{graphicx}
\usepackage{amsmath,amssymb}
\usepackage{algorithm} %%% for comm. to N-sensei
\usepackage{algpseudocode} %%% for comm. to N-sensei
\usepackage{color}
\usepackage{enumerate}
\begin{document}
\title{A Linear Delay Algorithm for Enumeration of $2$-Edge/Vertex-connected Induced Subgraphs\thanks{The work was partially supported by JSPS KAKENHI Grant Number 20K04978.}}
\titlerunning{Enumeration of 2-Edge/Vertex-connected Induced Subgraphs}
% If the paper title is too long for the running head, you can set
% an abbreviated paper title here
%
\author{Takumi Tada \inst{1} \and
Kazuya Haraguchi\inst{1}}
\authorrunning{T.~Tada and K.~Haraguchi}
% First names are abbreviated in the running head.
% If there are more than two authors, 'et al.' is used.
%
\institute{Graduate School of Informatics, Kyoto University, Japan}
\maketitle              % typeset the header of the contribution
\begin{abstract}
    For a set system $(V,{\mathcal C}\subseteq 2^V)$,
    we call a subset $C\in{\mathcal C}$ a {\em component}.
    A nonempty subset $Y\subseteq C$ is a
    {\em minimal removable set} ({\em MRS}) {\em of $C$}
    if $C\setminus Y\in{\mathcal C}$ and
    no proper nonempty subset $Z\subsetneq Y$ satisfies
    $C\setminus Z\in{\mathcal C}$.  
    In this paper, we consider the problem of enumerating
    all components in a set system such that, 
    for every two components $C,C'\in{\mathcal C}$ 
    with $C'\subsetneq C$, 
    every MRS $X$ of $C$ satisfies 
    either $X\subseteq C'$ or $X\cap C'=\emptyset$.  
    We provide a partition-based algorithm for this problem, 
    which yields the first linear delay algorithms 
    to enumerate all 2-edge-connected induced subgraphs, and
    to enumerate all 2-vertex-connected induced subgraphs.

\end{abstract}

\newcommand{\IDX}{{\mathit{Idx}}}
\newcommand{\MCE}{{\mathcal C}_e}
\newcommand{\MM}{{\mathcal M}}
\newcommand{\MS}{{\mathcal S}}
\newcommand{\MC}{{\mathcal C}}
\newcommand{\MCV}{{\mathcal C}_v}
\newcommand{\bbZ}{{\mathbb{Z}}}
\newcommand{\MRS}{{\textsc{Mrs}}}
\newcommand{\MAX}{{\textsc{Max}}}
\newcommand{\INV}{{\textsc{Inv}}}
\newcommand{\CAN}{{\textsc{Can}}}
\newcommand{\CANS}{{\textsc{Can}_{>2}}}
\newcommand{\CANP}{{\textsc{Can}_{=2}}}
\newcommand{\ART}{{\textsc{Art}}}
\newcommand{\CMPMRS}{{\textsc{ComputeMrs}}}

\newcommand{\NINT}{{\text{NonIntersect}}}
\newcommand{\NIL}{{\textsc{Nil}}}

\newcommand{\ANS}{{\mathit{Ans}}}
\newcommand{\SEQ}{{\mathit{Seq}}}
\newcommand{\DEP}{{\mathit{depth}}}

\newcommand{\secref}[1]{Sect.~\ref{sec:#1}}
\newcommand{\tabref}[1]{Table~\ref{tab:#1}}
\newcommand{\figref}[1]{Fig.~\ref{fig:#1}}
\newcommand{\corref}[1]{Corollary~\ref{cor:#1}}
\newcommand{\lemref}[1]{Lemma~\ref{lem:#1}}
\newcommand{\propref}[1]{Proposition~\ref{prop:#1}}
\newcommand{\thmref}[1]{Theorem~\ref{thm:#1}}
% newtheorem
\newtheorem{cor}{Corollary}
\newtheorem{defn}{Definition}
\newtheorem{lem}{Lemma}
\newtheorem{prop}{Proposition}
\newtheorem{thm}{Theorem}

\algrenewcommand{\algorithmicrequire}{\textbf{Input: }}
\algrenewcommand{\algorithmicensure}{\textbf{Output: }}

\long\def\invis#1{}
\section{Introduction}
\label{sec:intro}

An enumeration problem asks to output all feasible solutions
without duplication, whereas an optimization problem
asks to output just one feasible solution that is optimal with respect to
objective function. 
Enumeration is among significant issues
in such fields as discrete mathematics~\cite{Stanley.1999}, % enumerative combinatorics
algorithm theory~\cite{Avis.1996} % reverse search
and at the same time,
it has many practical applications in data mining
and bioinformatics~\cite{Agrawal.1993,Roux.2022,Sese.2010}. % copine, caldera, and so on. 

Given a graph, subgraph enumeration asks to list all subgraphs
that satisfy required conditions. It could find interesting substructures
in network analysis. 
Enumeration of cliques is among such problems~\cite{Chang.2013,Conte.2020,Makino.2004}, % Makino and Uno, and so on; survey recent research
where a clique is a subgraph such that
every two vertices are adjacent to each other
and thus may represent a group of SNS users
that are pairwise friends.
Pursuing further applications,
there have been studied enumeration of subgraphs
that satisfy weaker connectivity conditions,
e.g., pseudo-cliques~\cite{Uno.2010}. % gugure

In this paper, we consider enumeration of subgraphs
that satisfy weakest but fundamental connectivity conditions;
2-edge-connectivity and 2-vertex-connectivity.
For a positive integer $k$,
a graph is $k$-edge-connected (resp., $k$-vertex-connected)
if removal of any $k-1$ edges (resp., any $k-1$ vertices)
does not disconnect the graph and there
are $k$ edges (resp., $k$ vertices)
whose removal disconnects the graph. 
For the case of $k=2$,
we consider enumerating all vertex subsets
that induce 2-edge-connected or 2-vertex-connected subgraphs. 

For a graph $G$, let $V(G)$ and $E(G)$ denote
the set of vertices of $G$ and the set of edges of $G$,
respectively. An enumeration algorithm in general outputs
many solutions, and its {\em delay} refers
to computation time
between the start of the algorithm and the first output;
between any consecutive two outputs; and
between the last output and the halt of the algorithm. 
The algorithm attains {\em polynomial delay}
(resp., {\em linear delay})
if the delay is bounded by a polynomial (resp., a linear function)
with respect to the input size.

The main results of the paper are summarized in the following two theorems.

\begin{thm}
    \label{thm:ec}
    Let $G$ be a simple undirected graph,
    $n:=|V(G)|$ and $m:=|E(G)|$. 
    All $2$-edge-connected induced subgraphs of $G$
    can be enumerated
    in $O(n+m)$ delay and space.
\end{thm}

\begin{thm}
    \label{thm:vc}
    Let $G$ be a simple undirected graph,
    $n:=|V(G)|$ and $m:=|E(G)|$. 
    All $2$-vertex-connected induced subgraphs of $G$
    can be enumerated
    in $O(n+m)$ delay and space.
\end{thm}

%%%%%
We achieve the first linear delay algorithms
for enumerating 2-edge/vertex connected induced subgraphs. 
Ito et al.~\cite{ITO.2022} made the first study on
enumeration of 2-edge-connected induced subgraphs,
presenting a polynomial delay algorithm based on reverse search~\cite{Avis.1996} % Avis and Fukuda
such that the delay is $O(n^3m)$. 
Haraguchi and Nagamochi~\cite{Haraguchi.2022} % Algorithmica (2022)
studied an enumeration problem in a confluent set system
that includes enumeration of $k$-edge-connected (resp., $k$-vertex-connected)
induced subgraphs
as special cases, which yields $O(\min\{k+1,n\}n^5m)$
(resp., $O(\min\{k+1,n^{1/2}\}n^{k+4}m)$) delay algorithms.
Wen et al.~\cite{Wen.2019}
proposed an algorithm for enumerating maximal vertex subsets
that induce $k$-vertex-connected subgraphs
such that the total time complexity is 
$O(\min\{n^{1/2},k\}m(n+\delta(G)^2)n)$,
where $\delta(G)$ denotes the minimum degree over the graph $G$. 
%%%%%

We deal with the two subgraph enumeration problems in a more general framework. 
For a set $V$ of elements,
let $\MC\subseteq 2^V$ be a family of subsets of $V$.
A pair $(V,\MC)$ is called a {\em set system}
and a subset $C\subseteq V$ is called a {\em component} if $C\in\MC$. 
A nonempty subset $Y\subseteq C$ of a component $C$
is a {\em removable set of $C$}  if $C\setminus Y\in\MC$.
Further, a removable set $Y$ of $C$ is minimal,
or a {\em minimal removable set} ({\em MRS}),
if there is no $Z\subsetneq Y$ that is a removable set of $C$.
We denote by $\MRS_{\MC}(C)$ the family of all MRSs of $C$. 
Let us introduce the notion of SD property
of set system as follows. 

\begin{defn} A set system $(V,\MC)$ has {\em subset-disjoint} {\em (SD)}
  {\em property}
  if, for any two components $C,C'\in\MC$ such that $C\supsetneq C'$,
    either $Y\subseteq C'$ or $Y\cap C'=\emptyset$ holds
    for every MRS $Y$ of $C$. 
\end{defn}

% \begin{defn} A set system $(V,\MC)$ has {\em subset-disjoint} {\em (SD)}
%   {\em property}
%   if
%   \begin{enumerate}
%   \item $V\in\MC$; and
%   \item for any two components $C,C'\in\MC$ such that $C\supsetneq C'$,
%     either $Y\subseteq C'$ or $Y\cap C'=\emptyset$ holds
%     for every MRS $Y$ of $C$. 
%   \end{enumerate}
% \end{defn}

We consider the problem of enumerating all components
that are subsets of a given component
in a set system with SD property.
We assume that the set system is implicitly given by an oracle
such that, for a component $C\in\MC$ and a subset $X\subseteq C$,
the oracle returns an MRS $Y$ of $C$ that is disjoint with $X$ if exists;
and \NIL\ otherwise.
We denote the time and space complexity of the oracle by
$\theta_t$ and $\theta_s$, respectively. 
We show the following theorem that is 
a key for proving Theorems~\ref{thm:ec} and \ref{thm:vc}. 

\begin{thm}
  \label{thm:SD}
  Let $(V,\MC)$ be a set system with SD property,
  $C\in\MC$ be a component,
  $n:=|C|$ and $q:=\max\{|Y|\mid Y\in\MRS_{\MC}(C'),\ C'\in\MC,\ C'\subseteq C\}$.
  All components that are subsets of $C$ can be enumerated
  in $O(q+\theta_t)$ delay and $O(n+\theta_s)$ space
  with an exception that it takes $\Theta(n)$ time to output $C$
  as the first component. 
\end{thm}
% \begin{thm}
%   \label{thm:SD}
%   Let $(V,\MC)$ be a set system with SD property,
%   $n:=|V|$ and $q:=\max\{|Y|\mid Y\in\MRS_{\MC}(C),\ C\in\MC\}$.
%   Given a component $C\in\MC$, 
%   all components that are subsets of $C$ can be enumerated
%   in $O(q+\theta_t)$ delay and $O(n+\theta_s)$ space
%   with an exception that it takes $\Theta(n)$ time to output $C$
%   as the first component. 
% \end{thm}
% \begin{thm}
%   \label{thm:SD}
%   Let $(V,\MC)$ be a set system with SD property,
%   $n:=|V|$ and $q:=\max\{|Y|\mid Y\in\MRS_{\MC}(C),\ C\in\MC\}$.
%   All components in $\MC$ can be enumerated
%   in $O(q+\theta_t)$ delay and $O(n+\theta_s)$ space
%   with an exception that it takes $\Theta(n)$ time to output $V$
%   as the first component. 
% \end{thm}

The paper is organized as follows. 
After making preparations in \secref{prel},
we present an algorithm that enumerates
all components that are subsets of a given component in a set system with SD property,
along with complexity analyses in \secref{P.A},
as a proof for \thmref{SD}.
Then in \secref{P.EV.C},
we provide proofs for Theorems~\ref{thm:ec} and \ref{thm:vc}.
There are two core parts in the proofs.
In the first part, 
given a 2-edge-connected (resp., 2-vertex-connected) graph $G$,
we show that a set system $(V,\MC)$ should have SD property
if $V=V(G)$ and $\MC$ is the family of all vertex subsets
that induce 2-edge-connected (resp., 2-vertex-connected) subgraphs.
This means that 2-edge/vertex-connected induced subgraphs
can be enumerated by using the algorithm developed for \thmref{SD}. 
%Finally in \secref{proof},
Then in the second part, we explain how we design the oracle
to achieve linear delay and space.
% the oracle $\CMPMRS$ can be realized for such set systems
% so that $\theta_t=O(|V(G)|+|E(G)|)$ and $\theta_s=O(|V(G)|+|E(G)|)$,
%
We give concluding remarks in \secref{conc}.

%%%%%%%%%%%%%%%%%%%%%%%%%%%%%%%%%%%%%%%%%%%%%%%%%%%%%%%%%%%%
\section{Preliminaries}
\label{sec:prel}
Let $\bbZ$ and $\bbZ_+$ denote the set of integers
and the set of nonnegative integers, respectively.
For two integers $i,j\in\bbZ$ $(i\le j)$, let us denote 
$[i,j]:=\{i,i+1,\dots,j\}$.

For any sets $P,Q$ of elements,
when $P\cap Q=\emptyset$,
we may denote by $P\sqcup Q$
the disjoint union of $P$ and $Q$
in order to emphasize that they are disjoint.  

\subsection{Set Systems}
\label{sec:prel_set}

Let $(V,\MC)$ be a set system
which does not necessarily have SD property.
For any two subsets $U,L\subseteq V$,
we denote $\MC(U,L):=\{C\in\MC\mid L\subseteq C\subseteq U\}$.
Recall that, for a component $C\in\MC$,
we denote by $\MRS_{\MC}(C)$ the family of all MRSs of $C$.
Further, for $X\subseteq C$,
we denote $\MRS_{\MC}(C,X):=\{Y\in\MRS_{\MC}(C)\mid Y\cap X=\emptyset\}$.
For any two components $C,C'\in\MC$ with $C\supsetneq C'$,
let $Y_1,Y_2,\dots,Y_\ell\subseteq C\setminus C'$
be subsets such that $Y_i\cap Y_j=\emptyset$, $1\le i<j\le\ell$;
and $Y_1\sqcup Y_2\sqcup\dots\sqcup Y_\ell=C\setminus C'$
(i.e., $\{Y_1,Y_2,\dots,Y_\ell\}$ is a partition of $C\setminus C'$). 
Then $(Y_1,Y_2,\dots,Y_\ell)$ is
an {\em MRS-sequence} ({\em between $C$ and $C'$}) if
\begin{itemize}
\item $C'\sqcup Y_1\sqcup\dots\sqcup Y_i\in\MC$, $i\in[1,\ell]$; and
\item $Y_i\in\MRS_{\MC}(C'\sqcup Y_1\sqcup\dots\sqcup Y_i)$, $i\in[1,\ell]$.
\end{itemize}
One easily sees that there exists an MRS-sequence
for every $C,C'\in\MC$ such that $C\supsetneq C'$;
Suppose a sequence of components $(C_0:=C,C_1,C_2,\dots,C_{\ell-1},C_{\ell}:=C')$
such that $C_{i}$,
$i\in[1,\ell]$ is a maximal component among subsets of $C_{i-1}$
that contain $C_{\ell}=C'$ as a subset.
Then $(Y_1:=C_{\ell-1}\setminus C_{\ell},Y_2:=C_{\ell-2}\setminus C_{\ell-1},\dots,Y_\ell:=C_0\setminus C_1)$ is an MRS-sequence between $C$ and $C'$.

The following lemma holds regardless of SD property. 
\begin{lem}
    \label{lem:empty-abs}
    For any set system $(V,\MC)$,
    let $C\in\MC$ and
    $X\subseteq C$.
    It holds that
    $\MRS_{\MC}(C,X)=\emptyset\iff\MC(C,X)=\{C\}$.
\end{lem}
\begin{proof}
  $(\Longrightarrow)$
  %First we show $\MRS(C,C\setminus X)=\emptyset
  %\Rightarrow\MC(C,X)=\{C\}$.
  Suppose that there is a component $C'\in\MC(C,X)$ such that $C'\ne C$.
  Then $X\subseteq C'\subsetneq C$ holds,
  and there exists an MRS-sequence between $C$ and $C'$,
  say $(Y_1,Y_2,\dots,Y_\ell)$.
  We see that $Y_\ell\in\MRS_{\MC}(C)$ and $Y_\ell\subseteq C\setminus X$.
  Hence we have $Y_\ell\in\MRS_{\MC}(C,X)$ and thus $\MRS_{\MC}(C,X)\ne\emptyset$.
  $(\Longleftarrow)$
  Suppose that $\MRS_{\MC}(C,X)\neq\emptyset$.
  Let $Y\in\MRS_{\MC}(C,X)$, where we have $Y\subseteq C$ and $Y\cap X=\emptyset$.
  Then we see that $C\setminus Y$ is a component such that $C\supsetneq C\setminus Y\supseteq X$, indicating that $C\setminus Y\in \MC(C,X)$.
\qed\end{proof}

As we described in \secref{intro},
we assume that a set system $(V,\MC)$ with SD property 
is given implicitly by an oracle. 
We denote by $\CMPMRS_{\MC}$ the oracle.
Given a component $C\in\MC$ and a subset $X\subseteq C$
as a query to the oracle,
$\CMPMRS_{\MC}(C,X)$ returns one MRS in $\MRS_{\MC}(C,X)$
if $\MRS_{\MC}(C,X)\ne\emptyset$,
and \NIL\ otherwise,
where we denote by $\theta_t$ and $\theta_s$ 
the time and space complexity, respectively. 

The following lemma states 
a necessary condition of SD property
which is not sufficient. 
  
\begin{lem}
    \label{lem:disjoint-mrs-abs}
    Suppose that a set system $(V,\MC)$ with SD property is given.
    For every component $C\in\MC$,
    the minimal removable sets in $\MRS_{\MC}(C)$ are pairwise disjoint.
\end{lem}
\begin{proof}
  Suppose that there are $Y_1,Y_2\in\MRS_{\MC}(C)$
  such that $Y_1\cap Y_2\ne\emptyset$.
  By minimality, $Y_1\not\subseteq Y_2$ and $Y_1\not\supseteq Y_2$ hold,
  by which we have $Y_2\setminus Y_1\ne\emptyset$.
  Then $C\setminus Y_1$ is a component such that
  $(C\setminus Y_1)\cap Y_2 = (C\cap Y_2)\setminus(Y_1\cap Y_2)=Y_2\setminus(Y_1\cap Y_2)=Y_2\setminus Y_1\ne\emptyset$.
  It holds that $C\setminus Y_1\not\supseteq Y_2$ 
  since otherwise $Y_1\cap Y_2=\emptyset$ would hold.
  We have seen that neither $(C\setminus Y_1)\cap Y_2=\emptyset$ nor
  $C\setminus Y_1\supseteq Y_2$ holds, contradicting that
  $(V,\MC)$ has SD property. 
  \qed\end{proof}

\subsection{Graphs}

% Basic definition of a graph
\invis{
  For a graph $G$,
  we denote by $V(G)$ the vertex set of $G$
  and by $E(G)$ the edge set of $G$.
  The number of vertices of $G$ is called
  {\em order of $G$}
  and a graph of order $0$ or $1$ is called {\em trivial}.
}
Let $G$ be a simple undirected graph. 
For a vertex $v\in V(G)$, we denote by $\deg_G(v)$
the degree of $v$ in the graph $G$.
We let $\delta(G):=\min_{v\in V(G)}\deg_G(v)$. %notate the smallest degree among the vertices of $G$ as $\delta(G)$.
Let $S\subseteq V(G)$ be a subset of vertices.
\invis{
  %%% unused %%%
  We denote by $N_G(S)$ the union of neighbors of
  vertices in $S$, that is,
  $N_G(S)\triangleq\{v\in V(G)\mid\exists u\in S,
  uv\in E(G)\}\setminus S$.
}
A {\em subgraph induced by $S$}
is a subgraph $G'$ of $G$
such that $V(G')=S$ and $E(G')=\{uv\in E(G)\mid u,v\in S\}$
and denoted by $G[S]$.
%We denote by $\MAX_c(S)$ the set of connected components in $G[S]$.
For simplicity, we write 
the induced subgraph $G[V(G)\setminus S]$
as $G-S$. 
Similarly, for $F\subseteq E(G)$,
we write as $G-F$ the subgraph
whose vertex set is $V(G)$
and edge set is $E(G)\setminus F$.

% Connectivity
%%% Let $G$ be a simple undirected graph.
%
A {\em cut-set of $G$} is a subset $F \subseteq E(G)$
such that $G-F$ is disconnected.
In particular, we call an edge that
constitutes a singleton cut-set a {\em bridge}.
We define the {\em edge-connectivity $\lambda(G)$ of $G$}
to be the cardinality of the minimum cut-set of $G$
unless $|V(G)|=1$. %$G$ is a graph of order $1$.
If $|V(G)|=1$, %the order of $G$ is $1$,
then $\lambda(G)$ is defined to be $\infty$.
$G$ is called {\em $k$-edge-connected}
if $\lambda(G)\ge k$.
A {\em vertex cut} of $G$ is a subset $S \subseteq V(G)$
such that $G-S$ is disconnected.
In particular, we call a vertex cut whose size is two
a {\em cut point pair}
and a vertex that constitutes a singleton vertex cut
an {\em articulation point}.
For a subset $S\subseteq V(G)$,
let $\ART(S)$ denote the set of all articulation points
in $G[S]$.
We define the {\em vertex-connectivity $\kappa(G)$ of $G$}
to be the cardinality of the minimum vertex cut of $G$
unless $G$ is a complete graph.
If $G$ is complete,
then $\kappa(G)$ is defined to be $|V(G)|-1$.
$G$ is called {\em $k$-vertex-connected}
if $|V(G)|>k$ and $\kappa(G)\ge k$.
Obviously, $G$ is 2-edge-connected (resp., 2-vertex-connected)
if and only if there is no bridge (resp., no articulation point) in $G$. 

%
%(e.g., Proposition 1.4.2 in \cite{Diestel})
\begin{prop}[\cite{Wh.1932}]
  \label{prop:Whitney}
  Suppose that we are given a simple undirected graph $G$.
  If $|V(G)|\ge2$, then it holds that
  $\kappa(G)\le\lambda(G)\le\delta(G)$. 
\end{prop}

%%%%%%%%%%%%%%%%%%%%%%%%%%%%%%%%%%%%%%%%%%%%%%%%%%%%%%%%%%%%
\section{Enumerating Components in Set System with SD Property}
\label{sec:P.A}

In this section, we propose an algorithm
that  enumerates all components that are subsets of a given component in a set system $(V,\MC)$ with SD property
and conduct complexity analyses, as a proof for \thmref{SD}.

Let us introduce mathematical foundations
for a set system with SD property
that are necessary for designing our enumeration algorithm.

\begin{lem}
  \label{lem:all-abs}
    For a set system $(V,\MC)$ with SD property,
    let $C\in\MC$ be a component and
    $I\subseteq C$ be a subset of $C$.
    % such that $\MRS(C,C\setminus I)\not=\emptyset$.
    %
    %
    For any $Y\in\MRS_{\MC}(C,I)$, it holds that
    $\MC(C,I)=\MC(C\setminus Y,I)\sqcup\MC(C,I\sqcup Y)$. 
    %$\MC(C\setminus Y,I)\cap \MC(C,I\cup Y)=\emptyset$.
\end{lem}
\begin{proof}
  Two subsets $Y$ and $I$ are disjoint by definition of $\MRS_{\MC}(C,I)$. 
  $\MC(C\setminus Y,I)$ and $\MC(C,I\sqcup Y)$ are disjoint
  since every component in the former
  contains no element in $Y$, whereas
  every component in the latter contains $Y$ as a subset.
  For the equality, it is obvious that
  $\MC(C,I)\supseteq\MC(C\setminus Y,I)\sqcup\MC(C,I\sqcup Y)$ holds. 
  Let $C'\in\MC(C,I)$.
  If $C'=C$, then $C'\in\MC(C,I\cup Y)$.
  Otherwise (i.e., if $I\subseteq C'\subsetneq C$),
  either $Y\subseteq C'$ or $Y\cap C'=\emptyset$ holds.
  If $Y\subseteq C'$, then $C'\in\MC(C,I\cup Y)$ and
  otherwise, $C'\in\MC(C\setminus Y,I)$.
\qed\end{proof}

\begin{lem}
  \label{lem:valid-abs}
  For a set system $(V,\MC)$ with SD property,
  let $C\in\MC$ be a component,
  $I\subseteq C$ be a subset of $C$,
  and $\MRS_{\MC}(C,I):=\{Y_1,Y_2,\dots,Y_k\}$.
  It holds that
  \[
  \MC(C,I)=\{C\}\sqcup\big(\bigsqcup_{i=1}^k\MC(C\setminus Y_i,I\sqcup Y_1\sqcup\dots\sqcup Y_{i-1})\big). 
  \]
\end{lem}
\begin{proof}
  By \lemref{all-abs}, we have a partition $\MC(C,I)=\MC(C\setminus Y_1,I)\sqcup\MC(C,I\sqcup Y_1)$.
  We see that $\MRS_{\MC}(C,I\sqcup Y_1)=\{Y_2,\dots,Y_k\}$
  since $\MRS_{\MC}(C,I\sqcup Y_1)\subsetneq\MRS_{\MC}(C,I)$ and
  $Y_1\cap Y_i=\emptyset$, $i\in[2,k]$ holds by \lemref{disjoint-mrs-abs}.
  Applying \lemref{all-abs} recursively, we have
  \begin{align*}
    \MC(C,I)&=\MC(C\setminus Y_1,I)\sqcup\MC(C,I\sqcup Y_1)\\
    &=\MC(C\setminus Y_1,I)\sqcup\MC(C\setminus Y_2,I\sqcup Y_1)\sqcup\MC(C,I\sqcup Y_1\sqcup Y_2)\\
    &=\bigg(\bigsqcup_{i=1}^k
    \MC(C\setminus Y_i,
    I\sqcup Y_1\sqcup\dots\sqcup Y_{i-1})\bigg)\sqcup
    \MC(C,I\sqcup Y_1\sqcup\dots\sqcup Y_k),
  \end{align*}
  where $\MC(V,I\sqcup Y_1\sqcup\dots\sqcup Y_k)=\{C\}$
  holds by \lemref{empty-abs} since
  $\MRS_{\MC}(C,I\sqcup Y_1\sqcup\dots\sqcup Y_k)=\emptyset$.
\qed\end{proof}

\paragraph{Algorithm.}
Let $C\in\MC$, $I\subseteq C$ and $\MRS_{\MC}(C,I):=\{Y_1,Y_2,\dots,Y_k\}$. 
\lemref{valid-abs} describes a partition of $\MC(C,I)$
such that there exists a similar partition for
$\MC(C_i,I_i)$, where 
$C_i=C\setminus Y_i$ and
$I_i=I\sqcup Y_1\sqcup\dots\sqcup Y_{i-1}$, $i\in[1,k]$. 
Then we have an algorithm that enumerates all components in $\MC(C,I)$
by outputting $C$ and then outputting all components in
$\MC(C_i,I_i)$ for each $i\in[1,k]$
recursively. 

%{\color{red} Rough sketch of algorithm.}
For $C\in\MC$ and $I\subseteq C$,
Algorithm~\ref{alg:abs} summarizes a procedure to enumerate all
components in $\MC(C,I)$. 
Procedure $\textsc{List}$ in Algorithm \ref{alg:abs} outputs $C$ in line 2 and
computes $\MC(C_i,I_i)$, $i\in[1,k]$ recursively in line 6.
For our purpose, it suffices to invoke $\textsc{List}(C,\emptyset)$
to enumerate all components in $\MC(C,\emptyset)$.
% where $V$ is the only maximal component in the given set system.

%Algorithm \ref{alg:abs} calls $\textsc{List}(V,\emptyset)$ in line 1,
%and so the algorithm enumerates all components in $\MC$.

\begin{algorithm}[t!]
\caption{An algorithm to enumerate all components in $\MC(C,I)$, where $C\in\MC$ is a component in a set system $(V,\MC)$ with SD property and $I$ is a subset of $C$}
\label{alg:abs}
\begin{algorithmic}[1]

\Require A component $C\in\MC$ and a subset $I\subseteq C$
\Ensure All components in $\MC(C,I)$
\Procedure{$\textsc{List}$}{$C,I$}
%\Comment{enumerate $\MC(C,I)$}
\State Output $C$;
\State $X\gets I$;
\While{$\CMPMRS_{\MC}(C,X)\not=\NIL$}
    \State $Y\gets\CMPMRS_{\MC}(C,X)$;
    \State $\textsc{List}(C\setminus Y,X)$;
    \State $X\gets X\cup Y$
\EndWhile
\EndProcedure
\end{algorithmic}
\end{algorithm}

To prove \thmref{SD},
we introduce a detailed version of the algorithm
in Algorithm~\ref{alg:abs-while}.
% space complexity
We mainly use a stack to store data,
where we can add a given element (push), peek the element that is most recently added (last), remove the last element (pop), and 
shrink to a given size by removing elements that are most recently added (shrinkTo)
in constant time, respectively,
by using an array and managing an index to the last element.
In Algorithm \ref{alg:abs-while},
the set $X$ of Algorithm \ref{alg:abs} is realized by a stack.
In addition, a stack $\SEQ$ stores MRSs, a stack $\IDX$ stores the indices of the boundary of $X$ between parent and child, and an integer $\DEP$ represents the depth of recursion in Algorithm \ref{alg:abs}.
We can see that the algorithm enumerates all components that are subsets of $C$ in $\MC$ by \lemref{valid-abs}.
For better time complexity, we apply the {\em alternative method} to our algorithm in line 4--6 and 13--15 and reduce the delay~\cite{Uno.2003}.
We can also reduce the delay by outputting only difference from the previous output.
For a previous output $C'\in\MC$ and $Y\subseteq V$, we output "$+Y$" (resp., "$-Y$")
if the component is $C'\cup Y$ (resp., $C'\setminus Y$).
It takes $\Theta(|C|)$ time to output the first component $C$,
whereas it takes $O(q)$ time to trace differences in line 9 and 18,
where $q=\max\{|Y|\mid Y\in\MRS_{\MC}(C'),\ C'\in\MC(C,\emptyset)\}$.
% space complexity
We next discuss the space complexity.
The stack $\ANS$ contains the difference between components and thus it uses $O(q)$ space.
We see that $\SEQ$ stores MRS-sequence between a component and $C$ since
$\SEQ$ pops the MRS after traversing the child,
and so it consumes $O(n)$ space.
The stack $X$ uses $O(n)$ space by \lemref{disjoint-mrs-abs} and the definition of $\CMPMRS_{\MC}$,
and moreover $\IDX$ uses only $O(n)$ space since the maximum depth is $n$.

\paragraph{Proof for \thmref{SD}.}

%{\color{red} To achieve the required time and space complexity,}

%\begin{proof}
    % Let $C$ be a component and $X\subseteq C$ be the subset.
    % For $Y\in\MRS_{\MC}(C,X)$,
    % we compute the inputs $C':=C\setminus Y$ and $I':=X$
    % in the recursive call of $\textsc{List}$ in Algorithm \ref{alg:abs-while}.
    % At the end of the recursive call,
    % we store $C'$ and $X':=I'\cup\MRS_{\MC}(C',C'\setminus I')$
    % by \lemref{disjoint-mrs-abs}.
    % This indicates that we can restore $C$ and $X$
    % if we hold $Y$ and the size of $X$
    % since $I'\cap\MRS_{\MC}(C',C'\setminus I')=\emptyset$.
    % Algorithm \ref{alg:abs-while} store MRS in $\SEQ$ for 
    % each search node until the recursive call terminates and
    % restore the parent date using while-loop.

    We can see that Algorithm \ref{alg:abs-while}
    surely enumerates all components in $\MC(C,\emptyset)$
    by \lemref{valid-abs}.
    We first prove that Algorithm \ref{alg:abs-while} works
    in $O(q+\theta_t)$ delay except that
    it takes $\Theta(n)$ time to output $C$ as the first component.
    We should output $C$ as a first component, so that
    preparation (i.e., the statement "$\ANS$.push($C$)")
    needs $\Theta(n)$ time.
    We apply the alternative method to our algorithm, and thus
    we will see that operations in the while-loop
    are done in $O(q+\theta_t)$ time.
    A component can be outputted in $O(q)$ time
    by computing the difference from the previous one,
    except for the first component $C$.
    The difference can be traced by subtracting and adding
    a MRS before and after the depth changes,
    thus it takes $O(q)$ time.
    In addition, $\CMPMRS_{\MC}$ works in $\theta_t$ time 
    by definition and
    another operations are adding and subtracting a MRS,
    where computation time is $O(q)$.

    We next discuss
    the space complexity of Algorithm \ref{alg:abs-while}.
    We obtain that maximum size of the depth is $n$
    since the size of the input component is
    monotonically decreasing during while the depth increases
    and the termination condition that
    the MRS of the input component is initially empty
    is satisfied at most $n$ depth.
    The rest to show is that
    the space for $\SEQ$, $X$, and $\IDX$ are $O(n)$.
    For a component $C'\in\MC(C,\emptyset)$, we obtain that
    the $\SEQ$ is equivalent to the MRS-sequence
    between $C'$ and $C$ since $\SEQ$ store a new MRS
    before the depth increases and discards before
    the depth decreases.
    $X$ can be hold in $O(n)$ space since
    for any subset $I\subseteq C'$,
    $I$ and $\MRS_{\MC}(C',I)$ are pairwise disjoint.
    It is obvious that $\IDX$ uses $O(n)$ space
    since the maximum depth is $n$.
    We then obtain that whole space is $O(n+\theta_s)$ space.
    \hfill\qed
%\end{proof}

\begin{algorithm}[t!]
\caption{An algorithm to enumerate all components that are subsets of $C\in\MC$ in $(V,\MC)$ with SD property}
\label{alg:abs-while}
\begin{algorithmic}[1]

\Require A set system $(V,\MC)$ with SD property and a component $C\in\MC$
\Ensure All components that are subsets of $C$ in $\MC$

\State $\ANS,\SEQ,X,\IDX\gets$ empty stack; $C'\gets C$;
\State $\ANS$.push("$C$"); $\IDX$.push(0); $\DEP\gets 1$;
\While{$\DEP\not=0$}
    \If{$\DEP$ is odd}
        \State Output $\ANS$; $\ANS$.clear()
    \EndIf;
    \If{$\CMPMRS_{\MC}(C',X)\not=\NIL$}
    \Comment{emulate recursive call}
        \State $Y\gets\CMPMRS_{\MC}(C',X)$;
        \State $\ANS$.push("$-Y$");
        \State $\SEQ$.push($Y$); $C'\gets C'\setminus Y$; $\IDX$.push($X$.length());
        \State $\DEP\gets\DEP+1$
    \Else\Comment{trace back}
        \If{$\DEP$ is even}
            \State Output $\ANS$; $\ANS$.clear()
        \EndIf;
        \State $C'\gets C'\cup \SEQ$.last();
        \State $X$.shrinkTo($\IDX$.last()); $\IDX$.pop();
        \State $\ANS$.push("$+\SEQ$.last()");
        \State $X$.push($\SEQ$.last());
        \State $\SEQ$.pop(); $\DEP\gets\DEP-1$
    \EndIf
\EndWhile
\end{algorithmic}
\end{algorithm}

%%%%%%%%%%%%%%%%%%%%%%%%%%%%%%%%%%%%%%%%%%%%%%%%%%%%%%%%%%%%
\section{Enumerating 2-Edge/Vertex-Connected Induced Subgraphs}
\label{sec:P.EV.C}

In this section, we provide proofs for Theorems~\ref{thm:ec} and \ref{thm:vc}. 
Suppose that
we are given a simple undirected graph $G$
that is 2-edge-connected (resp., 2-vertex-connected).
In \secref{const}, we show that a set system $(V,\MC)$
has SD property if $V=V(G)$
and $\MC$ is the family of all vertex subsets
that induce 2-edge-connected subgraphs
(resp., 2-vertex-connected subgraphs).
This indicates that all 
2-edge/vertex-connected induced subgraphs
can be enumerated by the algorithm in the last section.
Then in \secref{linear}, we show how 
to design the oracle for generating an MRS
so that the required computational complexity is achieved. 

\subsection{Constructing Set Systems with SD Property}
\label{sec:const}

% system for 2-edge
We define
\begin{align*}
  \MCE&\triangleq\{C\subseteq V(G)\mid
  G[C]\text{ is }2\text{-edge-connected and }|C|>1\};\\
  \MCV&\triangleq\{C\subseteq V(G)\mid
  G[C]\text{ is }2\text{-vertex-connected}\}.
\end{align*}
Our problem is to enumerate all
subsets in $\MCE$ or $\MCV$.
We do not include singletons (that are 2-edge-connected components
by definition) in $\MCE$ in order to make the problem tractable in our formulation.
The following lemma is immediate by \propref{Whitney}.
\begin{lem}
  \label{lem:subset}
  For a simple undirected graph $G$,
  it holds that $\MCV\subseteq\MCE$.
\end{lem}

A {\em block} of a simple undirected
graph $G$ is a maximal connected subgraph
that has no articulation point.
Every block of $G$ is an isolated vertex;
a cut-edge (i.e., an edge whose removal increases the number of connected components); or a maximal $v$-component; e.g., see Remark 4.1.18 in \cite{We.2018}. 
The following lemma is immediate. 
\begin{lem}
    \label{lem:c-h}
    For a given simple undirected graph $G$,
    let $C\in\MCE$ be an e-component.
    It holds that $C=\bigcup_{H\in\MAX_v(C)}H$.
\end{lem}

We deal with a set system $(V,\MC)$
such that $V=V(G)$ and $\MC\in\{\MCE,\MCV\}$.
We use the notations and terminologies for set systems
that were introduced in \secref{prel_set}
to discuss the problem of enumerating subsets in $\MCE$ or $\MCV$.
For $S\subseteq V(G)$,
we call $S$ an {\em e-component} (resp., a {\em v-component})
if $S\in\MCE$ (resp., $S\in\MCV$).
By \lemref{subset}, every v-component is an e-component. 
Let $\MAX_e(S)$ (resp., $\MAX_v(S)$) denote the family of
all maximal e-components (resp., v-components)
among the subsets of $S$.
For an e-component $C\in\MCE$ (resp., a v-component $C\in\MCV$),
we write the family $\MRS_{\MCE}(C)$ (resp., $\MRS_{\MCV}(C)$)
of all minimal removable sets of $C$
as $\MRS_e(C)$ (resp., $\MRS_v(C)$) for simplicity.
We call an MRS in $\MRS_e(C)$ (resp., $\MRS_v(C)$)
an {\em e-MRS of $C$} (resp., a {\em v-MRS of $C$}). 
A minimal e-component $C\in\MCE$ induces a cycle (i.e., $G[C]$ is a cycle)
since no singleton is contained in $\MCE$ by definition,
and in this case, it holds that $\MRS_e(C)=\emptyset$.

For $P\subseteq S\subseteq V(G)$,
$P$ is called a {\em two-deg path in $G[S]$}
(or {\em in $S$} for short) if
$\deg_{G[S]}(u)=2$ holds for every $u\in P$.
In particular, $P$ is a
{\em maximal two-deg path in $S$}
if there is no two-deg path $P'$ in $S$
such that $P\subsetneq P'$.
It is possible that a maximal two-deg path consists of just one vertex. 
For an e-component $C\in\MCE$, 
we denote by $\CANP(C)$ the family of 
all maximal two-deg paths in $C$.
We also denote $\CANS(C):=\{\{v\}\mid v\in C,\deg_{G[C]}(v)>2\}$
and define $\CAN(C)\triangleq\CANP(C)\sqcup\CANS(C)$.
It is clear that
every vertex in $C$ belongs to either a maximal two-deg path in $\CANP(C)$
or a singleton in $\CANS(C)$,
where there is no vertex $v\in C$ such that $\deg_{G[C]}(v)\le 1$
since $G[C]$ is 2-edge-connected. 

Let $G$ be a simple undirected graph.
For an e-component (resp., a v-component),
the set of e-MRSs (resp., v-MRSs) is characterized by
\lemref{mrs-e} (resp., \lemref{mrs-v}),
where we see that $\MRS_v(C)\subseteq\MRS_e(C)$ holds
for every v-component $C\in\MCV$.

\begin{lem}[Observation 3 in \cite{ITO.2022}]
  \label{lem:mrs-e}
  For a simple undirected graph $G$, let $C\in\MCE$.
  It holds that 
  $\MRS_e(C)=\{Y\in\CAN(C)\mid C\setminus Y\in\MCE\}$. 
  %(e.g., see Observation 3 in \cite{ITO.2022}).
\end{lem}

\invis{
We denote $\MRS_v(C):=\{Y\in\MRS_e(C)\mid G[C\setminus Y]
\text{ does not have any articulation points}\}$.
The following lemma tells that for a v-component $C$,
it holds that $\MRS_v(C)=\{Y\subseteq C\mid Y
\text{ is a nonempty minimal subset of }C
\text{ such that }C\setminus Y\in\MCV\}$.
}

\begin{lem}
    \label{lem:mrs-v}
    For a simple undirected graph $G$, let $C\in\MCV$.
    It holds that $\MRS_v(C)=\{Y\in\CAN(C)\mid C\setminus Y\in\MCV\}$. 
    %%%It holds that $\MRS_v(C)=\{Y\in\MRS_e(C)\mid C\setminus Y\in\MCV\}$. 
    \invis{
      For a given simple undirected graph $G$,
      let $(V,\MCV)$ be a set system.
      For any v-component $C\in\MCV$,
      $\MRS_v(C)$ is a family of all MRSs of $C$.
    }
\end{lem}
\begin{proof}
    Let us denote $\MM:=\{Y\in\CAN(C)\mid C\setminus Y\in\MCV\}$. 
    %%%Let us denote $\MM:=\{Y\in\MRS_e(C)\mid C\setminus Y\in\MCV\}$. 
  
    ($\supseteq$)
    Let $Z\in\MM$. We see that $C\setminus Z\in\MCV$ holds,
    and thus $C\setminus Z\in\MCE$ holds by \lemref{subset}.
    By $Z\in\CAN(C)$ and \lemref{mrs-e}, $Z$ is an e-MRS of $C$. 
    No proper nonempty subset $Z'\subsetneq Z$ satisfies
    $C\setminus Z'\in\MCV$
    since otherwise $C\setminus Z'\in\MCE$ would hold,
    contradicting that $Z$ is an e-MRS of $C$. We see that $Z\in\MRS_v(C)$. 
  
    ($\subseteq$) 
    Let $Y\in\MRS_v(C)$
    and $C':=C\setminus Y$. 
    It is obvious that $C'\in\MCV$.
    We show that $Y\in\CAN(C)$.
    For the case of $|Y|=1$, let $u\in Y$ be a vertex.
    We see that $\deg_{G[C]}(u)\ge 2$ holds by \propref{Whitney} and thus $\{u\}\in\CAN(C)$ holds.
    %We also see that $C\setminus Y\in\MCE$ holds by \lemref{subset}.
    %and that $C$ is not a minimal e-component,
    %indicating that $G[C]$ is not a cycle.
    % Suppose that $Y\notin\CAN(C)$ for contradiction.
    % 
    Suppose that $|Y|>1$ holds.
    We see that $\deg_{G[C'\cup\{v\}]}(v)\le 1$ holds
    for any vertex $v\in Y$ since otherwise
    there would be a vertex $w\in Y$ such that $C'\cup\{w\}\in\MCV$ and $\{w\}\subsetneq Y$,
    which contradicts the minimality of $Y$.
    There is an $xy\in E(G[C])$ such that $x\in C'$ and $y\in Y$ by connectivity of $G[C]$.
    Let $z\in C'$ be a vertex with $x\not=z$;
    such $z$ exists since $C'$ is a v-component and thus $|C'|>2$. 
    By 2-vertex-connectivity of $G[C]$,
    there surely exists a shortest path $P$ between $y$ and $z$ in $G[C]-\{x\}$.
    Suppose visiting the vertices in the path $P$ from $y$,
    and let $z'\not=y$ denote the first vertex such that
    $\deg_{G[C'\cup\{z'\}]}(z')=1$.
    %and let $z'\not=y$ denote the first vertex from $y$ in $P$ such that $\deg_{G[C'\cup\{z'\}]}(z')=1$ holds.
      We denote by $P'\subseteq P$ the subset of vertices that are between $y$ and $z'$ in $P$.
      The path $P'$ is a maximal two-deg path in $C'\cup P'$.
      We see that $z'$ is not adjacent to $x$ by definition
      and thus $x$ cannot be an articulation point in $G[C'\cup P']$. 
      Then $C'\cup P'\in\MCV$ and $P'\subseteq Y$ would hold and
      in consequence $P'=Y$ holds by minimality of $Y$,
      so that $Y\in\CAN(C)$ holds.
  \qed\end{proof}

%%%%%%%%%%%%%%%%%%%%%%%%%%%%%%%%%%%%%%%%%%%%%%%%%%
\invis{
\begin{lem}
    \label{lem:mrs-v}
    For a given simple undirected graph $G$,
    let $(V,\MCV)$ be a set system.
    For any v-component $C\in\MCV$,
    $\MRS_v(C)$ is a family of all MRSs of $C$.
\end{lem}
\begin{proof}    
    ($\supseteq$)
    Let $Y\in\MRS_v(C)$ be a v-MRS.
    We see that $Y$ is a subset of $C$ such that
    $C\setminus Y\in\MCV$ by definition.
    Suppose that $Y'\subsetneq Y$ is a nonempty
    proper subset such that $C\setminus Y'\in\MCV$.
    We then obtain that there is a vertex
    $v\in Y\setminus Y'$ such that
    $\deg_{G[C\setminus Y']}(v)\le 1$ and it holds that $\kappa(G[C\setminus Y'])\le
    \delta(G[C\setminus Y'])\le 1$ by \propref{Whitney},
    which contradicts the definition of $Y'$.

 ($\subseteq$)
    Suppose that $Z$ is a MRS of $C$ such that
    $Z\not\in\MRS_v(C)$ holds.
    We denote $C'=C\setminus Z$.
    It holds $|Z|>1$ by the definition of $\MRS_v$.
    We obtain that there is no vertex $u\in Z$
    such that $\deg_{G[C'\cup\{u\}]}(u)>1$ since otherwise,
    $C'\cup\{u\}$ is obviously 
    v-component, where $Z\setminus\{u\}$ is a
    nonempty subset of $C$.
    Let $z\in Z$ be a vertex such that
    $\deg_{G[C'\cup\{z\}]}(z)=1$.
    There is a vertex $x\in C\setminus Z$ such that
    $xz$ is an edge in $G[C]$ by connectivity.
    Let $y\in C\setminus Z$ be a vertex with $x\not=y$.
    There is a shortest path $P$
    between $z$ and $y$ in $G[C]-\{x\}$
    since $G[C]$ is $2$-vertex-connected.
    Let $z'\in C\setminus Z$ denote
    the first vertex from $z$ in $P$
    such that $\deg_{G[C'\cup\{z'\}]}(z')=1$
    and $P'$ denote the subset of vertices that are
    between $z$ and $z'$ in the path $P$.
    Then $P'$ is clearly a maximal two-deg path in
    $G[C'\cup P']$ and $C'\cup P'\in\MCV$ holds
    since $z'$ is not adjacent to $x$ by definition.
    We then obtain that $Z':=Z\setminus P'$
    is a nonempty subset of $C$ such that
    $C\setminus Z'\in\MCV$, which contradicts
    the definition of $Z$.
\end{proof}

We then obtain that all e-components in $\MCE$
can be enumerated by applying algorithm
proposed in \secref{P.A}
to each maximal e-components in turn.
}

%\paragraph{Properties of $(V,\MCV)$.}
%%%%%%%%%%%%%%%%%%%%%%%%%%%%%%%%%%%%%%%%%%%%%%%%%%
If a simple undirected graph $G$
is 2-edge-connected (resp., 2-vertex-connected),
then the set system $(V(G),\MC_e)$
(resp., $(V(G),\MC_v)$) has SD property,
as shown in the following 
\lemref{SD-e} (resp., \lemref{SD-v}).
%
%We omit the proof for \lemref{SD-v}
%since it can be shown almost similarly to \lemref{SD-e}. 
% The proofs for the two lemmas are almost similar. 

\begin{lem}
  \label{lem:SD-e}
  For a simple undirected 2-edge-connected graph $G$,
  the set system $(V(G),\MC_e)$ has SD property. 
\end{lem}
\begin{proof}
  We see that $V(G)\in\MC_e$ since $G$ is 2-edge-connected. 
  Let $C,C'\in\MC_e$ be e-components
  such that $C\supsetneq C'$ and
  $Y\in\MRS_e(C)$ be an e-MRS of $C$.
  We show that either $Y\subseteq C'$ or $Y\cap C'=\emptyset$ holds.
  The case of $|Y|=1$ is obvious. 
  Suppose $|Y|>1$. By \lemref{mrs-e}, 
  $Y$ induces a maximal two-deg path in $G[C]$ such that
  for any $u\in Y$ it holds $\deg_{G[C]}(u)=2$.
  If $Y\not\subseteq C'$ and $Y\cap C'\ne\emptyset$,
  then there would be two adjacent vertices $v,v'\in Y$
  such that $v\in C\setminus C'$ and $v'\in C'$,
  where we see that $\deg_{G[C']}(v')\le 1$ holds.
  The $C'$ is an e-component and thus $|C'|\ge2$. 
  By \propref{Whitney}, we obtain
  $1\ge\delta(G[C'])\ge\lambda(G[C'])$,
  which contradicts that $C'\in\MCE$.
\qed\end{proof}

\begin{lem}
  \label{lem:SD-v}
  For a simple undirected 2-vertex-connected graph $G$,
  the set system $(V(G),\MC_v)$ has SD property. 
\end{lem}
\begin{proof}
  We see that $V(G)\in\MC_v$ since $G$ is 2-vertex-connected.
  Let $C,C'\in\MC_v$ be v-components
  such that $C\supsetneq C'$ and
  $Y\in\MRS_v(C)$ be a v-MRS of $C$.
  We show that either $Y\subseteq C'$ or $Y\cap C'=\emptyset$ holds.
  The case of $|Y|=1$ is obvious. 
  Suppose $|Y|>1$. By \lemref{mrs-v}, 
  $Y$ induces a maximal two-deg path in $G[C]$ such that
  for any $u\in Y$ it holds $\deg_{G[C]}(u)=2$.
  If $Y\not\subseteq C'$ and $Y\cap C'\ne\emptyset$,
  then there would be two adjacent vertices $v,v'\in Y$
  such that $v\in C\setminus C'$ and $v'\in C'$,
  where we see that $\deg_{G[C']}(v')\le 1$ holds.
  The $C'$ is a v-component and thus $|C'|>2$. 
  By \propref{Whitney}, we obtain
  $1\ge\delta(G[C'])\ge\kappa(G[C'])$,
  which contradicts that $C'\in\MCV$.
  \invis{
  $H\in\MCV(H,\emptyset)$ holds by definition.
  For any v-components $C,C'\in\MCV(H,\emptyset)$ and
  a MRS $Y\in\MRS_v(C)$,
  either $Y\subseteq C'$ or $Y\cap C'=\emptyset$ holds
  since $\MCV(H,\emptyset)\subseteq\MCE(H,\emptyset)$
  holds by \lemref{subset},
  $\MRS_v(C)\subseteq\MRS_e(C)$ by definition, and
  $(H,\MCE(H,\emptyset))$ is $\NINT$.
  }
  \qed\end{proof}

\invis{
We first show the following lemmas
that tell the relationship between
$\MAX_v$, $\ART$, and $2$-edge-connectivity.
We then prove the properties of $(V,\MCV)$.

\begin{lem}
    \label{lem:root-vc}
    For a simple undirected graph $G$,
    let $S\subseteq V(G)$.
    It holds that $\MAX_v(S)=
    \bigcup_{C\in\MAX_c(S\setminus\ART(S))}
    \big(C\cup(N_{G[S]}(C)\cap\ART(S))\big)$.
\end{lem}
\begin{proof}
    ($\supseteq$) Any subset $X$ in the family in
    the right hand is a v-component
    since $G[X]$ contains no articulation point.
    $X$ is also maximal; if not so,
    there would be $v\in S\setminus X$ such that
    $v\in N_{G[S]}(X)$.
    This $v$ is adjacent to only one vertex
    in $X\cap\ART(S)$,
    contradicting that $X$ is a v-component.
    
    ($\subseteq$) Let $Y$ be a subset in $\MAX_v(S)$.
    If $Y$ is not in the family of the right hand, then
    it holds either
    $Y\subsetneq X_1$ for some subset $X_1$ in the family
    of the right hand, or
    there are vertices
    $u\in X_1$, $v\in X_2$, and
    $a\in\ART(S)$
    such that the edges $ua$ and $va$ is in $E(G[Y])$
    for some distinct subset $X_1$ and $X_2$
    in the family of the right hand
    since $Y$ is connected.
    If the former is true, then it contradicts that
    $Y$ is a maximal subset
    since $X_1$ is a v-component.
    Otherwise, then $G[Y]$ has
    an articulation point $a$,
    which contradicts that $Y$ is a v-component.
    \qed\end{proof}
}

\invis{
\begin{lem}
    \label{lem:c-h}
    For a given simple undirected graph $G$,
    let $C\in\MCE$ be an e-component.
    It holds that
    $C=\bigcup_{H\in\MAX_v(C)}H$.
\end{lem}
  \begin{proof}
    We obtain $C\supseteq\bigcup_{H\in\MAX_v(C)}H$
    by definition.
    We prove the converse.
    Suppose $v\in C$ be an element
    such that is not in the right hand.
    It holds that
    $\bigcup_{H\in\MAX_v(C)}H=
    \bigcup\bigg(\bigcup_{C'\in\MAX_c(C\setminus\ART(C))}
    \big(C'\cup(N_{G[C]}(C')\cap\ART(C))\big)\bigg)$
    by \lemref{root-vc}, and so
    it holds that $v\in\ART(C)$.
    The articulation point $v$ in $\ART(C)$
    is  however adjacent to some vertices in
    $C'\in\MC(C\setminus\ART(C))$
    since otherwise, i.e., $v$ is adjacent
    only to another articulation point $u$,
    then the edge $vu$ is a bridge of $G[C]$.
\qed\end{proof}
}

%%%%%%%%%%%%%%%%%%%%%%%%%%%%%%%%%%%%%%%%%%%%%%%%%%%%%%%%%%%%
%\subsection{Proofs of Theorems~\ref{thm:ec} and \ref{thm:vc}}
\subsection{Computing MRSs in Linear Time and Space}
\label{sec:linear}

%We show that all e-components
%and v-components can be enumerated in linear delay and space
%with respect to the input graph size, using the algorithm in \secref{P.A}. 
%For a given simple undirected graph $G$,
%let $n:=|V(G)|$ and $m:=|E(G)|$.

%\paragraph{Complexity of Computing MRSs.}
Let $G$ be a simple undirected graph.
We describe how we compute an e-MRS of an e-component
in linear time and space.
The case of computing a v-MRS of a v-component
can be done almost analogously.

Specifically, for a given e-component $C\in\MCE$ 
%(resp., v-component $C\in\MCV$)
and subset $X\subseteq C$,
how we design the oracle $\CMPMRS_{\MCE}(C,X)$
so that it outputs {\em one} e-MRS $A\in\MRS_e(C,X)$
if $\MRS_e(C,X)\ne\emptyset$,
%(resp., {\em one} v-MRS $A\in\MRS_v(C,X)$ if $\MRS_v(C,X)\ne\emptyset$),
and \NIL\ %
%reports that no such e-MRS  
%(resp., no such v-MRS) exists
otherwise,
in linear time and space.  
In what follows, we derive a stronger result that
{\em all} e-MRSs in $\MRS_e(C,X)$
%(resp., {\em all} v-components in $\MRS_v(C,X)$)
can be enumerated in linear delay and space.

The scenario of the proof is as follows. 
%Observe that $\MRS_e(C,C\setminus X)\subseteq\{S\in\CAN(C)\mid S\cap X=\emptyset\}$ holds.
\begin{description}
\item[(1)]
  We show that, to enumerate e-MRSs in $\MRS_e(C,X)$,
  it suffices to examine %enumerate e-MRSs in
  $\MRS_e(S,X)$
  for each $S\in\MAX_v(C)$ respectively.
  %where $\MAX_v(C)=\{C\}$ holds when $C\in\MC_v\subseteq\MC_e$.
  This indicates that we may assume $C$ to be a v-component.
  It is summarized as \corref{assumption}, followed by \lemref{decomposable}.
\item[(2)] Using a certain auxiliary graph,
  we show that it is possible to output in linear time and space
  all candidates in $\CAN(C)$ that are e-MRSs of $C$ (\lemref{mrs});
  recall that all candidates of e-MRSs are contained in \CAN(C) by
  \lemref{mrs-e}. 
\end{description}

\begin{lem}
  \label{lem:decomposable}
  For a simple undirected 2-edge-connected graph $G$, 
  let $V:=V(G)$ and $Y\subsetneq V$ be any subset of $V$. 
  Then $Y$ is an e-MRS of $V$
  if and only if there is $S\in\MAX_v(V)$ such that
  \begin{enumerate}
  \item[\rm (i)] %$Y\subsetneq S$ holds;
    $Y\cap S'=\emptyset$ holds for every $S'\in\MAX_v(V)$ such that $S'\ne S$; and
  \item[\rm (ii)] $Y$ is either
    a path that consists of all vertices in $S$ except one %in $S\setminus\ART(V)$
    or 
    an e-MRS of $S$. 
  \end{enumerate}
\end{lem}
\begin{proof}
  For the necessity,
  every v-component $S\in\MAX_v(V)$ is an e-component.
  % By \lemref{dismrs-e},
  By the definition of SD property, %\lemref{disjoint-mrs-abs},
  either $Y\subsetneq S$ or $Y\cap S=\emptyset$ should hold.
  Suppose that there are two distinct v-components
  $S,S'\in\MAX_v(V)$ such that $Y\subsetneq S$ and $Y\subsetneq S'$.
  This leads to $|Y|=1$ since $1\ge |S\cap S'|\ge |Y|\ge 1$,
  where the first inequality holds by the fact that
  two blocks share at most one vertex (e.g., Proposition~4.1.19 in \cite{We.2018}). 
  Then $Y$ is a singleton that consists
  of an articulation point of $G$,
  contradicting that $Y$ is an e-MRS of $V$.
  There is at most one $S\in\MAX_v(V)$ that contains $Y$
  as a proper subset, and such $S$ surely exists
  since there is at least one v-component in $\MAX_v(V)$
  that intersects $Y$ by \lemref{c-h}, which shows (i). 
  To show (ii),
  suppose that $G[S]$ is a cycle.
  Then it holds that $|S\cap\ART(V)|=1$;
  if $|S\cap\ART(V)|\ge2$, then no singleton or path in $S$ is an e-MRS of $V$,
  contradicting that $Y\subsetneq S$;
  if $|S\cap\ART(V)|=0$, then $S$ is a connected component in $G$,
  contradicting that $G$ is connected. Then $Y$ should be the
  path in $S$ that consists of all vertices except the only articulation point. 
  Suppose that $G[S]$ is not a cycle.  
  Let $u,v\in S$ be two distinct vertices.
  We claim that every path between $u$ and $v$ should not visit
  a vertex out of $S$; if there is such a path, 
  then the union of v-components visited by the path
  would be a v-component containing $S$,
  contradicting the maximality of $S$. 
  In the graph $G-Y$,
  there are at least two edge-disjoint paths between
  any two vertices $u,v\in S-Y$.
  These paths do not visit any vertex out of $S$,
  and thus $S-Y$ is an e-component.
  It is easy to see that $Y\in\CAN(S)$. %showing (i). 

  For the sufficiency, 
  suppose that $G[S]$ is a cycle.
  There is no e-MRS of $S$ by definition.
  The set $Y$ should be a path that consists of all vertices in $S$
  except one, and by (i), the vertex that is not contained in $Y$
  should be an articulation point (which implies $|S\cap\ART(V)|=1$).
  Suppose that $G[S]$ is not a cycle.
  An e-MRS of $S$ exists since $S$ is a non-minimal e-component. 
  Let $Y$ be an e-MRS of $S$ that satisfies (i),
  that is, $Y$ contains no articulation points in $\ART(V)$. 
  In either case, it is easy to see that $V\setminus Y$ is an e-component
  and that $Y\in\CAN(V)$ holds,
  showing that $Y$ is an e-MRS of $V$.   
\qed\end{proof}

%%%%%%%%%% COROLLARY ON DECOMPOSABILITY %%%%%%%%%%
\begin{cor}
  \label{cor:assumption}
  For a given simple undirected graph $G$,
  let $C\in\MCE$ be an e-component. 
  Then it holds that
  \begin{align*}
    \MRS_e(C)&=\big(\bigsqcup_{S\in\MAX_v(C):\ G[S]\textrm{\ is\ not\ a\ cycle}}\{Y\in\MRS_e(S)\mid Y\cap\ART(C)=\emptyset\}\big)\\
    &\sqcup\big(\bigsqcup_{S\in\MAX_v(C):\ G[S]\textrm{\ is\ a\ cycle\ and\ }|S\cap\ART(C)|=1}(S\setminus\ART(C))\big).
  \end{align*}
  %where the family of subsets on the right hand
  %is pairwise disjoint. 
  \invis{
    For a given simple undirected connected graph $G$,
    let $C\in\MCE$ be an e-component,
    $A \subsetneq C$ be a subset of $C$,
    and $H_1,H_2,\dots,H_r\in\MAX_v(C)$
    be all maximal v-components of $G[C]$.
    $A$ is an e-MRS of $C$ if and only if,
    one of the following conditions holds
    for some integer $s\in[1,r]$.

    \begin{enumerate}
        \item $G[H_s]$ is a cycle such that contains
        only one articulation point $v_a$ of $G[C]$ and
        $A=H_s\setminus\{v_a\}$.
        \item $A$ is a minimal removable set of $H_s$ and
        does not contain any articulation point of $G[C]$.
    \end{enumerate}
  }
\end{cor}
By the corollary,
to obtain e-MRSs of an e-component $C$,
it suffices to examine all maximal v-components
in $\MAX_v(C)$ respectively.

%%%%%%%%%% AUXILIARY GRAPHS %%%%%%%%%%
We observe the first family in the right hand in \corref{assumption}.
Let $C$ be a v-component
such that $G[C]$ is not a cycle.
For each path $P\in\CANP(C)$,
there are exactly two vertices $u,v\in \CANS(C)$
such that $u$ is adjacent to one endpoint of $P$
and $v$ is adjacent to the other endpoint of $P$.
We call such $u,v$ {\em boundaries of $P$}.
We denote the pair of boundaries of $P$ by $B(P)$,
that is, $B(P):=uv$.
%%% We define the union of $P\cup B(P)$ over all $P\in\CANP(C)$
%%% to be $\Lambda_{=2}(C)\triangleq\{P\cup B(P)\mid P\in\CANP(C)\}$.
We define $\Lambda_{>2}(C)\triangleq\{uv\in E(G)\mid u,v\in\CANS(C)\}$.
Let $\Lambda(C):=\CANP(C)\sqcup\Lambda_{>2}(C)$. 
We then define an auxiliary graph $H_C$ so that
\begin{align*}
  V(H_C):=&\CANS(C)\sqcup\CANP(C)\sqcup\Lambda_{>2}(C)\\
  =&\CANS(C)\sqcup\Lambda(C)=\CAN(C)\sqcup\Lambda_{>2}(C),\\
  E(H_C):=&\{uP\subseteq V(H_C)\mid u\in\CANS(C),\ P\in\CANP(C),\ u\in B(P)\}\\ & \sqcup
  \{ue\subseteq V(H_C)\mid u\in \CANS(C),\ e\in\Lambda_{>2}(C),\ u\in e\}. 
\end{align*}
We call a vertex in $\CANS(C)$ an {\em ordinary} vertex, whereas
we call a vertex in $\Lambda(C)=\CANP(C)\sqcup\Lambda_{>2}(C)$
an {\em auxiliary} vertex.

The auxiliary graph $H_C$ is obtained by
editing $G[C]$ as follows; 
replace each path $P\in\CANP(C)$ with a single vertex
(which is an auxiliary vertex in $\CANP(C)$); and then
insert a vertex into each edge that joins two vertices in $\CANS(C)$
(which is an auxiliary vertex in $\Lambda_{>2}(C)$). 
One readily sees that $H_C$ is 2-vertex-connected. 

\invis{
See \figref{aux} for an example of
an auxiliary graph of a given v-component $C$.
In (a), we see that $C$ consists of 15 vertices.
The set $\CANS(C)$ consists of eight black round vertices,
whereas $\CANP(C)$ consists of four paths given by seven white round vertices,
where one out of the four is a singleton. 
%We also see that there exist nine edges
%that connects ordinary vertices.
%We then connects ordinary vertices and
%auxiliary vertices according to the definition.
%We can easily see that auxiliary graph is a graph
%such that candidates in $\CANP(V(G))$ are contracted
%(illustrated as white triangle) and
%auxiliary vertices are inserted into the edge
%that connects ordinary vertices (illustrated as white square).
%
The auxiliary graph $H_C$ is shown in (b),
where $\CANS(C)$ again consists of black round vertices,
$\CANP(C)$ consists of white triangle vertices, and
$\Lambda_{>2}(C)$ consists of white square vertices. 
}
See \figref{aux} for an example of
an auxiliary graph of a given v-component $C$.
In (a), we see that $C$ consists of 16 vertices.
The set $\CANS(C)$ consists of nine black round vertices,
whereas $\CANP(C)$ consists of four paths given by seven white round vertices,
where one out of the four is a singleton. 
The auxiliary graph $H_C$ is shown in (b),
where $\CANS(C)$ again consists of black round vertices,
$\CANP(C)$ consists of white triangle vertices, and
$\Lambda_{>2}(C)$ consists of white square vertices. 

\begin{figure}[t!]
  \centering
  \begin{tabular}{clc}
    \includegraphics[width=4.5cm]{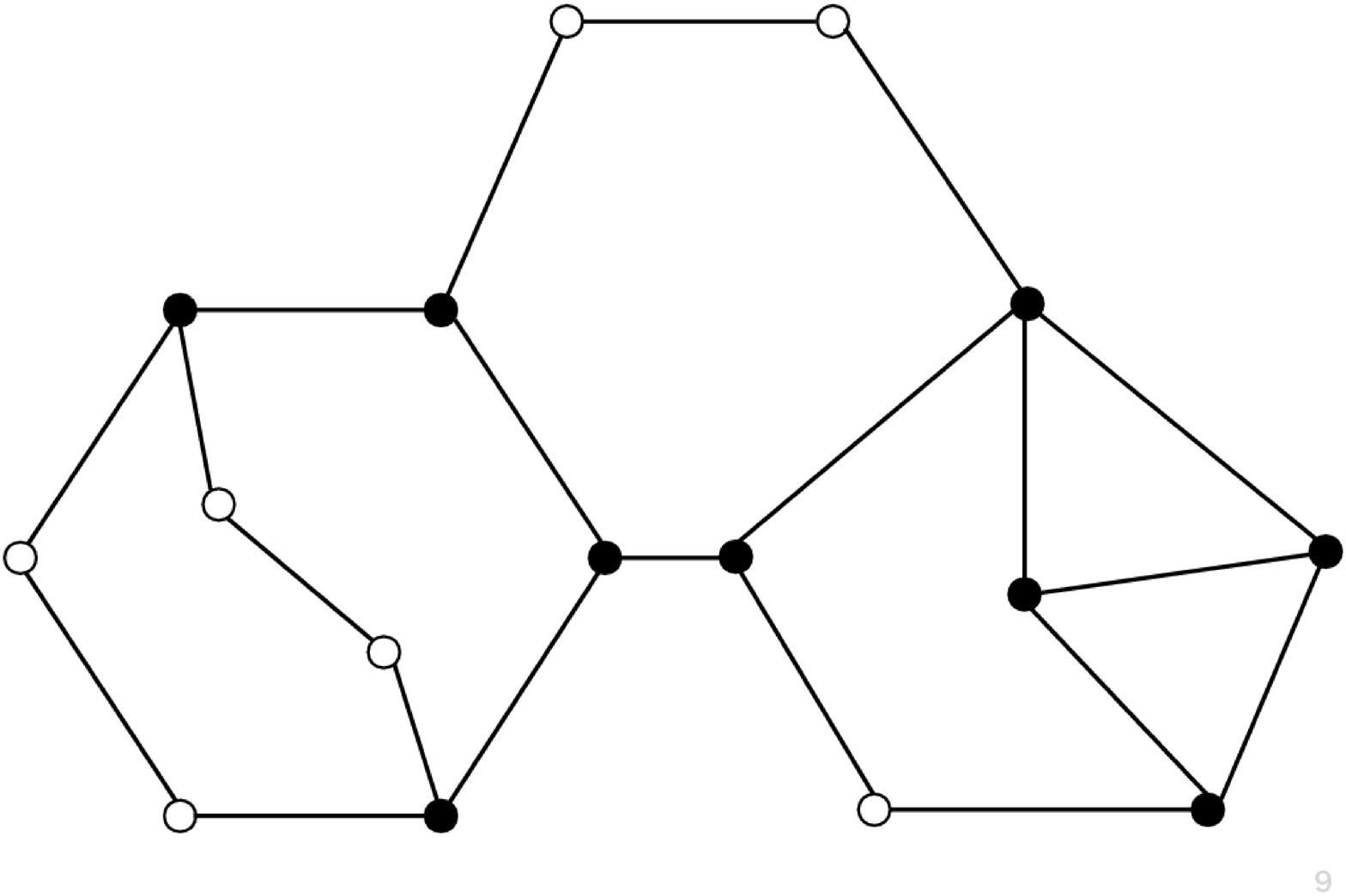}&
    \ \ \ &
    \includegraphics[width=4.5cm]{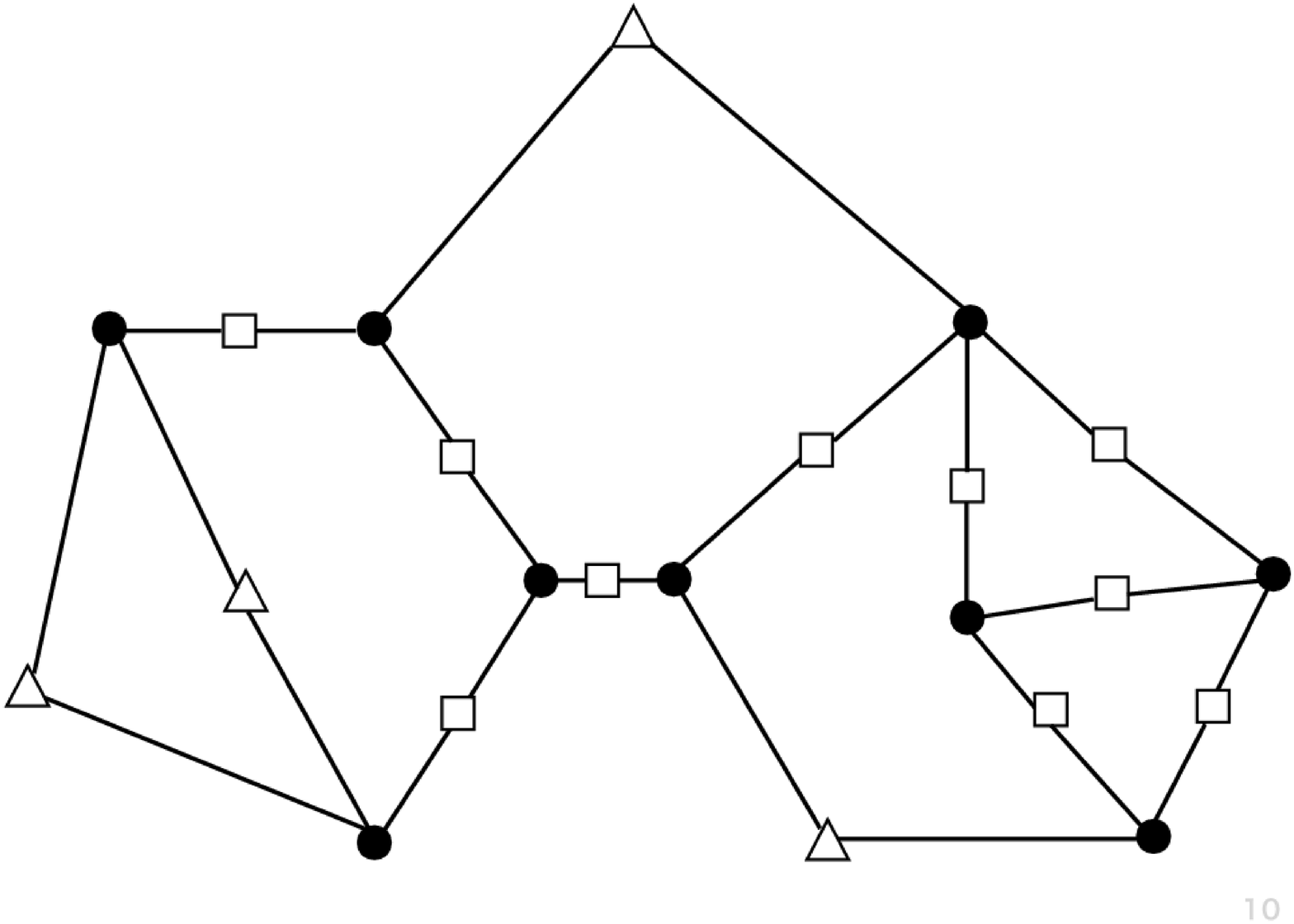}\\
    (a) && (b)
  \end{tabular}
  \caption{(a) v-component $C$; (b) auxiliary graph $H_C$}
  \label{fig:aux}
\end{figure}

For $P\in\CANP(C)$, we denote by $E(P)$ the set of all edges in
the path $P\cup B(P)$.
For $e\in\Lambda_{>2}(C)$, we denote $E(e):=\{e\}$.
%%% for some $u,v\in\CANS(C)$.
We see that $E(G[C])=\bigsqcup_{h\in\Lambda(C)}E(h)$ holds.
% where $\biguplus$ denotes the disjoint union. 

\begin{lem}
  \label{lem:uve}
  Given a simple undirected graph $G$,
  let $C\in\MC_v$ be a v-component such that $G[C]$ is not a cycle
  and $Y\in\CAN(C)$.
  Then $Y\in\MRS_e(C)$ holds if and only if
  there is no auxiliary vertex $h\in\Lambda(C)$ 
  such that $\{Y,h\}$ is a cut point pair of $H_C$. 
\end{lem}
\begin{proof}
  For the necessity, suppose that there is $h\in\Lambda(C)$
  such that $\{Y,h\}$ is a cut point pair of $H_C$.
  Then $h$ is an articulation point of $H_C-Y$. 
  Every edge $e\in E(h)$ is a bridge in $G[C]-Y$,
  indicating that $G[C]-Y$ is not 2-edge-connected,
  and hence $Y\notin\MRS_e(C)$.

  For the sufficiency, suppose that $Y\notin\MRS_e(C)$.
  Then $G[C]-Y$ is not 2-edge-connected but should be connected.
  There exists a bridge, say $e$, in $G[C]-Y$. 
  Let $h\in\Lambda(C)$ be the auxiliary vertex such that $e\in E(h)$.
  We see that $h$ is a cut point of $H_C-Y$, indicating that
  $\{Y,h\}$ is a cut point pair of $H_C$. 
\qed\end{proof}

%%%%%%%%%% LINEAR TIME ENUMERATION OF MRS %%%%%%%%%%
\begin{lem}
  \label{lem:mrs}
  Suppose that a simple undirected 2-edge-connected graph $G$ is given.
  Let %$n:=|V(G)|$, $m:=|E(G)|$ and
  $V:=V(G)$.
  For any subset $X\subseteq V$,
  all e-MRSs in $\MRS_e(V,X)$ can be enumerated
  in $O(n+m)$ time and space.
\end{lem}
\begin{proof}
  We can complete the required task as follows. 
  (1) We obtain $\ART(V)$ and decompose $V$ into maximal v-components.
  For each maximal v-component $C$,
  (2) if $G[C]$ is a cycle and $|C\cap\ART(V)|=1$, then output $C\setminus\ART(V)$ if it is disjoint with $X$;
  and (3) if $G[C]$ is not a cycle, then we construct an auxiliary graph $H_C$,
  compute all cut points pairs of $H_C$,
  and output all $Y\in\CAN(C)$ that are disjoint with $X\cup\ART(V)$
  and that are not contained
  in any cut point pair together with an auxiliary vertex. 
  The correctness of the algorithm follows by \corref{assumption}
  and \lemref{uve}.

  For the time complexity, (1) can be done in $O(n+m)$ time~\cite{Tarjan.1972}.
  For each $C\in\MAX_v(V)$, let $n_C:=|C|$ and $m_C:=|E(G[C])|$. 
  We can decide in $O(n_C+m_C)$ time
  whether $C$ is in (2), (3) or neither of them. 
  If we are in (2), then the task can be done in $O(n_C)$ time.
  If we are in (3), then the task can be done in $O(n_C+m_C)$ time
  since $H_C$ can be constructed in linear time and
  all cut point pairs of a 2-vertex-connected graph $H_C$
  can be enumerated in linear time~\cite{GutMut.2000,HopTar.1973}.
  An articulation point $v$ appears in at most $\deg_G(v)$
  maximal v-components, and hence $\sum_{C\in\MAX_v(V)}O(n_C)=O(n+m)$.
  The number of maximal v-components is $O(n)$,
  and the overall time complexity over $C\in\MAX_v(V)$
  is $O(n)+\sum_{C\in\MAX_v(V)}O(n_C+m_C)=O(n+m)$. 
  The space complexity analysis is analogous. 
  \invis{
    Suppose $G[C]$ is $2$-vertex-connected and not a cycle.
    This assumption does not lose generality by \corref{assumption}
    since all articulation points in $G$
    can be listed in $O(n+m)$\cite{Tarjan.1972}.
    In order to enumerate all MRSs of $C$,
    first we enumerate candidates
    $\CAN(C)$,
    and then we determine whether each candidate is in $\MRS_e(C)$
    using \corref{first} and \corref{second}.
    Candidates $\CAN(C)$ can be enumerated in $O(n+m)$ by
    checking its degree for each vertices.
    Next, we determine whether each candidate is in $\MRS_e(C)$ or not in $O(n+m)$ by \corref{first} and \corref{second}.
    It is obvious that $G[C]^\ast$ can be constructed
    in $O(n+m)$ time.
    We can enumerate all cut point pairs in $G^*[C]$
    since the all cut point pairs in $2$-vertex-connected graph can be listed
    in $O(n+m)$ time \cite{HopTar.1973, GutMut.2000} and
    thus we gain all MRSs of $V$ in $O(n+m)$ time.
    Then we compute $\MRS_e(C,C\setminus I)$ by
    checking for each minimal removable set containing
    a element of $I$ or not.
  }
\qed\end{proof}

Let us proceed to how we generate all v-MRSs of a v-component $C\in\MCV$. 
%%% By $\MCV\subseteq\MCE$,
%%% \corref{assumption} is applicable to $C$,
%%% where $\MAX_v(C)=\{C\}$ in this case.
%%% It holds that $\MRS_v(C)=\{Y\in\MRS_e(C)\mid G[C-Y]\textrm{\ has\ no\ articulation\ point}\}\subseteq\MRS_e(C)$. 
%
%%% If $G[C]$ is a cycle, then we are done
%%% since $\emptyset=\MRS_e(C)\supseteq\MRS_v(C)$.
%
In this case, the following lemma is
an analogue of \lemref{uve} for the e-component case.

\begin{lem}
  \label{lem:third}
  Given a simple undirected graph $G$,
  let $C\in\MC_v$ be a v-component and $Y\in\CAN(C)$.
  Then $Y\in\MRS_v(C)$ holds if and only if
  there is no ordinary vertex $u\in\CANS(C)$ 
  such that $\{Y,u\}$ is a cut point pair of $H_C$. 
  \invis{
    For a given simple undirected graph $G$,
    let $D$ be a v-component and
    suppose $G^*[D]$ is constructed from $G[D]$.
    For a candidate $X\in\CAN(D)$,
    let $v_X$ denote a vertex such that
    $G^*[D](X)=\{v_X\}$.
    Then $G[D]-X$ has an articulation point
    if and only if,
    for some non-auxiliary vertex $v$, 
    $\{v_X,v\}$ is the cut point pair of $G^*[D]$.
  }
\end{lem}
\begin{proof}
  For the necessity, suppose that there is $u\in\CANS(C)$
  such that $\{Y,u\}$ is a cut point pair of $H_C$.
  Then $u$ is an articulation point of $H_C-Y$.
  This means that $u$ is an articulation point of $G[C]-Y$.
  The subgraph $G[C]-Y$ is not 2-vertex-connected,
  and thus $Y\notin\MRS_v(C)$ holds.

  For the sufficiency, suppose that $Y\notin\MRS_v(C)$.
  There is an articulation point $u\in C-Y$
  of the subgraph $G[C]-Y$.
  If $u\in\CANS(C)$, then we are done.
  Otherwise (i.e., if $u\in\CANP(C)$),
  then let $u'$ be a boundary of the path that contains $u$.
  We see that $u'$ is an articulation point of $G[C]-Y$,
  and that $\{Y,u'\}$ is a cut point pair of $H_C$.   
  \invis{
    We first prove the sufficiency.
    Let $v_a$ be an articulation point of $G[D]-X$.
    There exist the candidates $Y,Z\in\CAN(D)$ and
    the vertices $y\in Y$ and $z\in Z$
    such that
    $y$ and $z$ are disconnected in $G[D]-X-\{v_a\}$.
    We denote by $v_Y\in G^*[D](Y)$ and
    $v_Z\in G^*[D](Z)$
    the associated vertices.
    We then obtain
    $v_Y$ and $v_Z$ are disconnected
    in $G^*[D]-\{v_X,v_a\}$
    since $v_a$ is also
    a non-auxiliary vertex in $G^*[D]$.
    We next prove the necessity.
    Let $v_b$ be a non-auxiliary vertex in $G^*[D]$
    such that $\{v_X,v_b\}$ is
    the cut point pair of $G^*[D]$
    and $v_y'$ and $v_z'$ be vertices that are
    disconnected in $G^*[D]-\{v_X,v_b\}$.
    There are distinct candidates
    $Y',Z'\in\CAN(D)$ such that
    $G^*[D](Y')=\{v_y'\}$ and $G^*[D](Z')=\{v_z'\}$
    by the definition of $\CAN$.
    We then obtain the vertices $y'\in Y'$ and $z'\in Z'$
    is disconnected in $G[D]-X-\{v_b\}$
    since $v_b$ is also a vertex in $G[D]$,
    so that $G[D]-X$ has an articulation point $v_b$.
  }
 \qed\end{proof}

Similarly to the e-component case,
we can show the following lemma on
complexity of generating v-MRSs of a v-component. 
\begin{lem}
    \label{lem:mrsvc}
    Suppose that a simple undirected 2-vertex-connected graph $G$ is given.
    Let %$n:=|V(G)|$, $m:=|E(G)|$ and
    $V:=V(G)$.
    For any subset $X\subseteq V$,
    all v-MRSs in $\MRS_v(V,X)$ can be enumerated
    in $O(n+m)$ time and space.
  \invis{
    For a given simple undirected graph $G=(V,E)$,
    let $D\subseteq V$ be a v-component and
    $I \subseteq D$ be a subset of $D$.
    Then, $\MRS_v(D,D\setminus I)$ can be enumerated
    in $O(n+m)$ time.
  }
\end{lem}
\invis{
\begin{proof}
    We denote by $C=D$ an e-component.
    We can find all MRS of $C$ in $O(n+m)$ time
    by \lemref{mrs}.
    Then we can compute all the MRS of $D$ by
    % \lemref{mrs-vc} 
    the definition of $MRS_v$
    and \lemref{third}.
\qed\end{proof}
}

%%%%%%%%%% PROOF FOR THE THEOREM %%%%%%%%%%
\paragraph{Proofs for Theorems~\ref{thm:ec} and \ref{thm:vc}.}
For \thmref{ec},
we see that $\MCE=\bigsqcup_{S\in\MAX_e(V)}\MCE(S,\emptyset)$. 
We can enumerate all maximal e-components in $\MAX_e(V)$
in $O(n+m)$ time and space, by removing all bridges in $G$. 
All e-components in $\MCE(S,\emptyset)$ for each $S\in\MAX_e(V)$
can be enumerated in $O(q+\theta_t)$ delay and in
$O(n+\theta_s)$ space by \thmref{SD},
where $q=|S|\le n$. 
We can implement $\CMPMRS_{\MCE}$ so that $\theta_t=O(n+m)$ and
$\theta_s=O(n+m)$ by \lemref{mrs}. 
\thmref{vc} is analogous. 
\hfill\qed

%%%%%%%%%%%%%%%%%%%%%%%%%%%%%%%%%%%%%%%%%%%%%%%%%%%%%%%%%%%%
\section{Concluding Remarks}
\label{sec:conc}
In this paper, we proposed the first linear delay algorithms
for enumerating all 2-edge/vertex-connected induced subgraphs
in a given simple undirected graph,
where the key tool is SD property of set system. 

The future work includes
extension of our framework to $k$-edge/vertex-connectivity for $k>2$; and
studying relationship between SD property
and set systems known in the literature (e.g., independent system, accessible system, strongly accessible system, confluent system). 

%Three future works:
%larger k 
%relation between set systems
%C_e and C_v

%%%%%%%%%%%%%%%%%%%%%%%%%%%%%%%%%%%%%%%%%%%%%%%%%%%%%%%%%%%%
% For citations of references, we prefer the use of square brackets
% and consecutive numbers. Citations using labels or the author/year
% convention are also acceptable. The following bibliography provides
% a sample reference list with entries for journal
% articles~\cite{ref_article1}, an LNCS chapter~\cite{ref_lncs1}, a
% book~\cite{ref_book1}, proceedings without editors~\cite{ref_proc1},
% and a homepage~\cite{ref_url1}. Multiple citations are grouped
% \cite{ref_article1,ref_lncs1,ref_book1},
% \cite{ref_article1,ref_book1,ref_proc1,ref_url1}.
%
% ---- Bibliography ----
%
% BibTeX users should specify bibliography style 'splncs04'.
% References will then be sorted and formatted in the correct style.
%
% \bibliographystyle{splncs04}
% \bibliography{mybibliography}

\clearpage
\begin{thebibliography}{8}
% data mining
\bibitem{Agrawal.1993}
Agrawal, R., Imieli\'{n}ski, T., Swami, A.:
Mining association rules between sets of items in large databases. In: Proceedings of the 1993 ACM SIGMOD International Conference on Management of Data, SIGMOD ’93, pp. 207--216 (1993).
\doi{10.1145/170035.170072}
% reverse search
\bibitem{Avis.1996}
Avis, D., Fukuda, K.:
Reverse search for enumeration.
Discrete Applied Mathematics \textbf{63}(1-3), 21--46 (1996).
\doi{10.1016/0166-218X(95)00026-N}
% maximal clieque 2
\bibitem{Chang.2013}
Chang, L., Yu, J.X., Qin, L.:
Fast Maximal Cliques Enumeration in Sparse Graphs.
Algorithmica \textbf{66}(1), 173--186 (2013).
\doi{10.1007/s00453-012-9632-8}
% maximal clique 3
\bibitem{Conte.2020}
Conte, A., Grossi, R., Marino, A., Versari, L.:
Sublinear-Space and Bounded-Delay Algorithms for Maximal Clique Enumeration in Graphs.
Algorithmica \textbf{82}(6), 1547--1573 (2020).
\doi{10.1007/s00453-019-00656-8}
%\bibitem{Diestel}
%Diestel, R.: Graph Theory (Graduate Texts in Mathematics). Springer (2005)
% SPQR-tree
\bibitem{GutMut.2000}
Gutwenger, C., Mutzel, P.:
A Linear Time Implementation of SPQR-Trees. In: Marks, J. (eds) Graph Drawing. GD 2000. Lecture Notes in Computer Science, vol 1984. Springer, Berlin, Heidelberg.
\doi{10.1007/3-540-44541-2\_8}
% support-closed subsets
\bibitem{Haraguchi.2022}
Haraguchi. K., Nagamochi, H.:
Enumeration of Support-Closed Subsets in Confluent Systems.
Algorithmica \textbf{84}, 1279--1315 (2022).
\doi{10.1007/s00453-022-00927-x}
% finding cut point pairs
\bibitem{HopTar.1973}
Hopcroft, J. E., Tarjan, R. E.:
Dividing a Graph into Triconnected Components.
SIAM Journal on Computing
\textbf{2}(3), 135--158, (1973).
\doi{10.1137/0202012}
% enumerating 2-edge-connected induced subgraphs
\bibitem{ITO.2022}
Ito, Y., Sano, Y., Yamanaka, K., Hirayama, T.:
A Polynomial Delay Algorithm for Enumerating 2-Edge-Connected Induced Subgraphs.
IEICE Transactions on Information and Systems
\textbf{E105.D}(3), 466--473 (2022).
\doi{10.1587/transinf.2021FCP0005}
% maximal clique
\bibitem{Makino.2004}
Makino, K., Uno, T.:
New Algorithms for Enumerating All Maximal Cliques. In: Hagerup, T., Katajainen, J. (eds)
Algorithm Theory - SWAT 2004. SWAT 2004. Lecture Notes in Computer Science, vol 3111, pp. 260--272.
Springer, Berlin, Heidelberg (2004).
\doi{10.1007/978-3-540-27810-8\_23}
% CALDERA
\bibitem{Roux.2022}
Roux de Bézieux, H., Lima, L., Perraudeau, F., Mary, A., Dudoit, S., Jacob, L.:
CALDERA: finding all significant de Bruijn subgraphs for bacterial GWAS.
Bioinformatics. \textbf{38}(1), 36--44 (2022).
\doi{10.1093/bioinformatics/btac238}
% COPINE
\bibitem{Sese.2010}
Seki, M., Sese, J.:
Identification of active biological networks and common expression conditions.
In: 8th IEEE International Conference on BioInformatics and BioEngineering,
BIBE 2008, pp. 1--6 (2008).
\doi{10.1109/BIBE.2008.4696746}.
\bibitem{Stanley.1999}
Stanley, R., Fomin, S.:
Enumerative Combinatorics (Cambridge Studies in Advanced Mathematics).
Cambridge University Press, Cambridge (1999)
\doi{10.1017/CBO9780511609589}
% finding articulation points
\bibitem{Tarjan.1972}
Tarjan, R.:
Depth-First Search and Linear Graph Algorithms.
SIAM Journal on Computing \textbf{1}(2), 146--160 (1972).
\doi{10.1137/0201010}
% finding bridges
\bibitem{Tarjan.1974}
Tarjan, R.:
A note on finding the bridges of a graph.
Information Processing Letters \textbf{2}(6), 160--161 (1974).
% alternative output
\bibitem{Uno.2003}
Uno, T.:
Two General Methods to Reduce Delay and Change of Enumeration Algorithms.
NII Technical Reports(2003)
% pseudo clique
\bibitem{Uno.2010}
Uno, T.: An Efficient Algorithm for Solving Pseudo Clique Enumeration Problem.
Algorithmica \textbf{56}, 3--16 (2010).
\doi{10.1007/s00453-008-9238-3}
% maximal k-vertex
\bibitem{Wen.2019}
Wen, D., Qin, D., Zhang, Y., Chang, L., Chen, L.:
Enumerating k-Vertex Connected Components in Large Graphs. 2019 IEEE 35th International Conference on Data Engineering (ICDE), pp. 52--63, Macao, China (2019).
\doi{10.1109/ICDE.2019.00014}
% block decomposition (Remark 4.1.18)
\bibitem{We.2018}
  West, D.B.: Introduction to Graph Theory (2nd Edition; reissue). Pearson Modern Classic (2018)
%Diestel, R.: Graph Theory (Graduate Texts in Mathematics). Springer (2005)
% kappa\le \lambda\le \delta
\bibitem{Wh.1932}
  Whitney, H. Congruent Graphs and The Connectivity of Graphs.
  Amer. J. Math. \textbf{54}, 150--168 (1932).
\end{thebibliography}
%

\end{document}